\documentclass[11pt]{article}
\usepackage[T1]{fontenc}
\usepackage{amsfonts}
\usepackage{amsmath}
\usepackage{amssymb}
\usepackage{amsthm}
\usepackage{bbm}
\usepackage{bm}
\usepackage{mathrsfs}
\usepackage{verbatim}
\usepackage{setspace}
\usepackage{color}
\usepackage{pdfsync}

\theoremstyle{plain}
\newtheorem{theorem}{Theorem}[section]

\newtheorem{proposition}[theorem]{Proposition}
\newtheorem{lemma}[theorem]{Lemma}

\theoremstyle{definition}
\newtheorem{definition}[theorem]{Definition}
\newtheorem{remark}[theorem]{Remark}

\newtheorem{example}[theorem]{Example}

\theoremstyle{remark}

\renewenvironment{thebibliography}[1]{%
\begin{oldthebibliography}{#1}%
\setlength{\baselineskip}{.9em}
\linespread{1}
\small
\setlength{\parskip}{0.2ex}%
\setlength{\itemsep}{.3em}%
}%
{%
\end{oldthebibliography}%
}

\newcommand{\eps}{\varepsilon}

\newcommand{\N}{\mathbb{N}}

\newcommand{\R}{\mathbb{R}}

\newcommand{\cA}{\mathcal{A}}
\newcommand{\cB}{\mathcal{B}}

\newcommand{\cD}{\mathcal{D}}

\newcommand{\cF}{\mathcal{F}}

\newcommand{\cM}{\mathcal{M}}

\newcommand{\cP}{\mathcal{P}}

\newcommand{\fP}{\mathfrak{P}}

\DeclareMathOperator{\graph}{graph}

\DeclareMathOperator{\NA2}{NA_{2}}
\DeclareMathOperator{\PCE}{PCE}

\DeclareMathOperator{\Int}{int}

\newcommand{\as}{\mbox{-a.s.}}
\newcommand{\qs}{\mbox{-q.s.}}
\newcommand{\1}{\mathbf{1}}

\newcommand{\br}[1]{\langle #1 \rangle}
\numberwithin{equation}{section}

\usepackage[pdfborder={0 0 0}]{hyperref}
\hypersetup{
  urlcolor = black,
  pdfauthor = {Bruno Bouchard, Marcel Nutz},
  pdfkeywords = {Transaction Costs; Arbitrage of Second Kind; Consistent Price System; Model Uncertainty},
  pdftitle = {Consistent Price Systems under Model Uncertainty},
  pdfsubject = {Consistent Price Systems under Model Uncertainty},
  pdfpagemode = UseNone
}

\begin{document}

\title{\vspace{-0em}
\mbox{Consistent Price Systems under Model Uncertainty}
\date{\today}
\author{
     Bruno Bouchard%
     \thanks{CEREMADE, Universit\'e Paris Dauphine and CREST-ENSAE, bouchard@ceremade.dauphine.fr. Research supported by ANR Liquirisk and Investissements d'Avenir (ANR-11-IDEX-0003/Labex Ecodec/ANR-11-LABX-0047).
     }
     \and
  Marcel Nutz%
  \thanks{
  Depts.\ of Statistics and Mathematics, Columbia University, mnutz@columbia.edu. Research supported by NSF Grant DMS-1208985.
  }
 }
}
\maketitle 

\begin{abstract} We develop a version of the fundamental theorem of asset pricing for discrete-time markets with proportional transaction costs and model uncertainty. A robust notion of no-arbitrage of the second kind is defined and shown to be equivalent to the existence of a collection of strictly consistent price systems.
\end{abstract}

\vspace{3em}

{\small

\vspace{1em}

}

\section{Introduction}\label{se:intro}

In the setting of frictionless financial market models, the fundamental theorem of asset pricing states that the absence of arbitrage opportunities is equivalent to the existence of an equivalent martingale measure. The nontrivial implication is to pass from the former to the latter, and is usually obtained by a Hahn--Banach separation argument in the context of $L^{p}$-spaces, as in the celebrated Kreps--Yan theorem~\cite{Yan.80}; we refer to \cite{DelbaenSchachermayer.06,FollmerSchied.04} for background. A similar reasoning still applies in the presence of proportional transaction costs as in \cite{KabanovRasonyiStricker.02,Schachermayer.04}; see also \cite{KabanovSafarian.09} for a monograph account. Clearly, such duality arguments depend on the presence of a reference probability measure.

In the setting of model uncertainty, the market is defined with respect to a family $\cP$ of (typically singular) probability measures. Each element is  interpreted as a possible model and arbitrage is defined in a robust way with respect to $\cP$. Different notions of arbitrage and related results have been discussed in the literature where transactions are costless, among them \cite{AcciaioBeiglbockPenknerSchachermayer.12, BouchardNutz.13, BurzoniFrittelliMaggis.14, DavisHobson.07, DolinskySoner.12, FernholzKaratzas.11, Riedel.11}. In particular, \cite{BouchardNutz.13} defines an arbitrage as a trade which can be entered for free and leads to a terminal wealth which is nonnegative but not vanishing in the $\cP$-quasi-sure sense. In this context, the classical duality theory does not apply because of the lack of a reference measure, and it seems that direct adaptations lead to somewhat abstract generalizations of martingale measures or fairly stringent compactness assumptions (e.g., \cite{DeparisMartini.04,Nutz.13}). This difficulty is avoided in \cite{BouchardNutz.13} by applying finite-dimensional separation arguments in a local fashion, similarly as in the original proof of the Dalang--Morton--Willinger theorem~\cite{DalangMortonWillinger.90}, and then passing to the multi-period setting by selection arguments. Thus, it is crucial that, for the no-arbitrage condition of \cite{BouchardNutz.13}, absence of local arbitrage is equivalent to absence of global arbitrage. 

An extension of \cite{BouchardNutz.13} to models with (proportional) transaction costs was recently proposed in~\cite{BayraktarZhang.13}, where quasi-sure versions of the classical \emph{weak} and \emph{robust} no-arbitrage conditions of~\cite{KabanovStricker.01trans} and~\cite{Schachermayer.04} are used. As in \cite{BouchardNutz.13}, the authors start with a local argument for one-period models. However, due to the transaction costs, absence of local arbitrage is no longer equivalent to absence of global arbitrage in their sense. Using an intricate forward-backward construction,~\cite{BayraktarZhang.13} nevertheless succeeds in passing from the one-period to the multi-period case. However, as mentioned by the authors, this method necessitates rather stringent continuity conditions on the price process and the set-valued process that governs the structure of $\cP$, cf.\ \cite[Assumption 2.1]{BayraktarZhang.13}, and its scope seems to be limited to markets with a single risky asset.       
 
The purpose of this paper is to introduce a different generalization of the frictionless no-arbitrage condition, one which naturally lends itself to the setting of transaction costs under model uncertainty and allows for a general version of the fundamental theorem of asset pricing. We consider a robust version of \emph{no-arbitrage of the second kind,} a condition first used by \cite{Rasonyi.09} in the context of transaction costs under the name ``no sure gain in liquidation value''; see also~\cite[Section~3.2.6]{KabanovSafarian.09} as well as \cite{BouchardTaflin.13} and the references therein. Our version postulates that a position cannot be $\cP$-quasi-surely solvent tomorrow if it is not already solvent today; cf.\  Definition~\ref{def: NA2} for the precise statement. The same condition can be stated in a global fashion, in terms of terminal positions, but the local and global notions remain equivalent as can be seen immediately (Remark~\ref{rem: local vs global}). Thus, the local-to-global philosophy can be applied naturally in our framework, and this leads to our version of the fundamental theorem  of asset pricing. 

Namely, assuming nontrivial transaction costs (\emph{efficient friction} in the terminology of \cite{KabanovRasonyiStricker.02}), we show that robust no-arbitrage of the second kind is equivalent to the existence of a family of strictly consistent price systems which are ``strictly consistently extendable'' in the language of~\cite{Rasonyi.09}. This family is rich enough to have the same polar sets as $\cP$; in particular, we retrieve the main result of~\cite{Rasonyi.09} when $\cP$ is a singleton. See Definition~\ref{de:PCE} and Theorem~\ref{th:FTAP} for the precise statements which are conveniently formulated in the language of Kabanov's general model of solvency cones; it includes, for instance, Schachermayer's model of exchange rates~\cite{Schachermayer.04}, see also Example~\ref{ex:pi}.

A different fundamental theorem was obtained in \cite{DolinskySoner.13} as a corollary of a superhedging result. In \cite{DolinskySoner.13}, there is a single risky asset with constant proportial transaction cost, represented by the canonical process on path space, and all paths are considered possible models; in addition, options are traded statically. In our multidimensional setting, options can be modeled dynamically like any other asset, using corresponding transaction costs.

The remainder of this paper is organized as follows. In Section~\ref{se:main}, we describe the setting in detail and state the main result together with the pertinent definitions. The proof is reported in Section~\ref{se:proof}; we first consider a single-period market and then pass the multi-period setting. Some notions and results of measure theory are collected in the Appendix for the convenience of the reader.

\section{Main Result}\label{se:main}

\subsection{Probabilistic Structure}

Let $T\in\N$ and let $\Omega_1$ be a Polish space. For $t\in\{0,1,\dots,T\}$, let $\Omega_t:=\Omega_1^t$ be the $t$-fold Cartesian
product, with the convention that $\Omega_0$ is a singleton, and let $\cF_t$ be the universal completion of  the Borel $\sigma$-field $\cB(\Omega_t)$. We write $(\Omega,\cF)$ for $(\Omega_T,\cF_T)$; this will be our basic measurable space and we shall often see $(\Omega_t,\cF_t)$ as a subspace of $(\Omega,\cF)$. We denote by $\fP(\Omega)$ the set of all probability measures on $\cB(\Omega)$ (or equivalently on $\cF$), equipped with the topology of weak convergence. For each $t\in\{0,1,\dots,T-1\}$ and $\omega\in\Omega_t$, we are given a nonempty convex set
$\cP_t(\omega)\subset \fP(\Omega_1)$, the possible models for the $t$-th period given the state $\omega$ at time $t$. We assume that for each $t$,
\[
  \graph(\cP_t):=\{(\omega,P):\, \omega\in \Omega_t,\, P\in\cP_t(\omega)\}\subset \Omega_t \times \fP(\Omega_t)\quad \mbox{is analytic.}
\]
This ensures that $\cP_t$ admits a universally measurable selector (cf.\ the Appendix); that is, a universally measurable kernel $P_t:\, \Omega_t\to \fP(\Omega_1)$ such that $P_t(\omega)\in \cP_t(\omega)$ for all $\omega\in\Omega_t$. If we are given such a kernel $P_t$ for each $t\in\{0,1,\dots T-1\}$, we can define a probability $P$ on $\Omega$ by Fubini's theorem,
\[
  P(A)=\int_{\Omega_1} \cdots \int_{\Omega_1} \1_A(\omega_1,\dots,\omega_T) P_{T-1}(\omega_1,\dots,\omega_{T-1}; d\omega_T)\cdots P_0(d\omega_1),
\]
where we write $\omega=(\omega_1,\dots,\omega_T)$ for a generic element of $\Omega\equiv\Omega_1^T$. The above formula will be
abbreviated as $P=P_0\otimes P_1\otimes \cdots \otimes P_{T-1}$ in the sequel. We can then introduce the set $\cP\subset\fP(\Omega)$ of models for the market up to time~$T$,
\[
  \cP:=\{P_0\otimes P_1\otimes \cdots \otimes P_{T-1}:\, P_t(\cdot)\in \cP_t(\cdot),\,t=0,1,\dots, T-1\},
\]
where each $P_t$ is a universally measurable selector of $\cP_t$. 

\subsection{Market Model}
We first introduce a general market model formulated in terms of random cones (cf.\ the Appendix for the pertinent definitions). In Example~\ref{ex:pi} we shall see how these cones arise in a typical example and how the technical conditions are naturally satisfied.

For each $t\in\{0,\ldots,T\}$, we consider a Borel-measurable random set $K_{t}: \Omega_{t}\to 2^{\R^{d}}$ . Each $K_{t}$ is a closed, convex cone containing $\R^{d}_{+}$, called the \emph{solvency cone} at time $t$. It represents the positions that are solvent; i.e., can be turned into nonnegative ones by immediate exchanges at current market conditions, see also Example~\ref{ex:pi} below. We denote by  $K^{*}_{t}$ its (nonnegative) dual cone,
$$
  K^{*}_{t}(\omega):=\{y\in \R^{d}: \br{x,y}\ge 0\mbox{ for all }x\in K_{t}(\omega)\}, \quad \omega\in \Omega_{t};
$$
here $\br{x,y}=\sum_{i=1}^{d} x^{i}y^{i}$ is the usual inner product on $\R^{d}$. We remark that $K^{*}_{t}\subset \R^{d}_{+}$ as $\R^{d}_{+}\subset K_{t}$ and that $K^{*}_{t}$ is again Borel-measurable; cf.\ Lemma~\ref{le:randomSets}. We assume that $K^{*}_{t}\cap \partial \R^{d}_{+}=\{0\}$, so every vector in $K^{*}_{t}\setminus\{0\}$ has strictly positive components. Moreover, we assume that
\begin{equation}\label{eq:assIntKstar}
 \Int K^{*}_{t}\ne \emptyset,\quad t\in\{0,\ldots,T\}.
\end{equation}
This condition is called \emph{efficient friction} since it embodies the presence of nontrivial transaction costs; cf.\ Example~\ref{ex:pi}. Finally, we assume that there is a constant $c>0$ such that
\begin{equation}\label{eq:KstarBound}
  x^{j}/y^{j} \leq c (x^{i}/y^{i}),\quad 1\leq i,j\leq d, \quad x,y\in K^{*}_{t}(\omega)\setminus\{0\}
\end{equation}
for all $\omega\in\Omega_{t}$ and $t\in\{0,\ldots,T\}$. This is equivalent to the seemingly weaker requirement that the above holds for all $x\in K^{*}_{t}(\omega)\setminus\{0\}$ and one fixed $y\in K^{*}_{t}(\omega)\setminus\{0\}$, which corresponds to a choice of numeraire. It is also equivalent to assume that~\eqref{eq:KstarBound} holds only for $i=1$. One can think of the division by $y$ as a normalization by a numeraire, and if we consider $d=2$ for simplicity, \eqref{eq:KstarBound} means that the angles between $K^{*}_{t}\subset \R^{2}_{+}$ and the coordinate axes are bounded away from zero after the normalization. In financial terms, \eqref{eq:KstarBound} means that the proportion of transaction costs is bounded from above, as we will see momentarily.

\begin{example}\label{ex:pi}
  We consider a market with $d$ assets (e.g., currencies) that is described by a matrix $(\pi^{ij}_{t})_{1\leq i,j\leq d}$ at every time $t$. The interpretation is that at time $t$, one can exchange $a\pi^{ij}_{t}$ units of asset $i$ for $a$ units of asset $j$ (where $a\geq0$); thus, $\pi^{ij}_{t}$ is the exchange rate between the assets, including transaction costs, and it is natural to assume that $\pi^{ij}_{t}$ is adapted and positive. This so-called model with physical units includes models with bid-ask spread or models with frictionless price and transaction cost coefficients; cf.\  \cite{KabanovSafarian.09,Schachermayer.04}.
  
  In the present setup, the solvency cone $K_{t}$ consists of those portfolios $x=(x^{1},\dots,x^{d})$ such that after a suitable exchange between the assets, the position in every asset is nonnegative; in formulas,
  $$
   K_{t}:=\bigg\{x\in \R^{d}:\, \exists\; (a^{ij})_{ij}\in \R_{+}^{d\times d}\mbox{ s.t. } x^{i}+\sum_{j\ne i} a^{ji}- a^{ij} \pi^{ij}_{t}\ge 0,\;\;1\leq i\leq d\bigg\}.
  $$
  This is also the convex cone generated by $\{e_{i}, \pi^{ij}_{t}e_{i}-e_{j}:\, 1\leq i,j\leq d\}$, where $(e_{i})_{i\le d}$ denotes the standard basis of $\R^{d}$. Thus, the dual cone is given by
  $$
    K^{*}_{t}=\big\{y\in \R^{d}_{+}:\, y^{j}\leq y^{i}\pi^{ij}_{t}, \;\;1\leq i,j\leq d\big\}.
  $$
  We may assume without loss of generality that $\pi^{ij}_{t}\leq \pi^{ik}_{t}\pi^{kj}_{t}$, meaning that direct exchanges are not more expensive than indirect ones.

  Let us now discuss the conditions introduced above. It is clear that $\R^{d}_{+}\subset K_{t}$ and $K^{*}_{t}\cap \partial \R^{d}_{+}=\{0\}$ are always satisfied. Moreover, if $\pi^{ij}_{t}: \Omega_{t}\to\R$ is Borel-measurable, it follows that $K_{t}$ is also Borel-measurable.
  Let us assume that
  \begin{equation}\label{eq:posTransCost}
    \pi^{ij}_{t}\pi^{ji}_{t}>1,\quad 1\leq i\neq j \leq d,
  \end{equation}
  meaning that there is a nonvanishing transaction cost for a round-trip between two assets. This condition is equivalent to $\Int K^{*}_{t}\ne \emptyset$, which was our requirement of efficient friction in~\eqref{eq:assIntKstar}. Indeed, by an elementary calculation, \eqref{eq:posTransCost} is equivalent to the angle between any two of the vectors spanning $K_{t}$ being strictly less than $180$ degrees, and this is in turn equivalent to $\Int K^{*}_{t}\ne \emptyset$. 
  Finally, let us assume that
 $$
   \pi^{1i}_{t}\pi^{i1}_{t}\leq c
 $$ 
 for a constant $c>0$, meaning that the cost of a round-trip is bounded---more precisely, no more that $c$ units of the first asset are necessary to end up with one unit of the same asset after the round-trip. 
 As $\pi^{ij}_{t}\leq \pi^{i1}_{t}\pi^{1j}_{t}$, this is equivalent to the symmetric condition  $\pi^{ij}_{t}\pi^{ji}_{t}\leq c$, $1\leq i\neq j \leq d$ (possibly after replacing $c$ by $c^{2}$) and we shall see that it is also equivalent to~\eqref{eq:KstarBound}. Indeed, let $S_{t}$ take values in $K^{*}_{t}\setminus\{0\}$; one may think of $S^{i}_{t}$ as a frictionless price for asset $i$, denominated in an external currency. Since $S_{t}$ then has positive components, we may introduce another matrix $(\lambda^{ij}_{t})$ via
 $$
   S^{j}_{t}(1+\lambda^{ij}_{t}) = \pi^{ij}_{t} S^{i}_{t}.
 $$
 Thus, $\lambda^{ij}_{t}$ can be interpreted as the proportional transaction cost incurred when exchanging asset $i$ for asset $j$; note that $S_{t}\in K^{*}_{t}$ implies $\lambda^{ij}_{t}\geq0$. Our condition that $\pi^{ij}\pi^{ji}\leq c$ yields that
 $$
   1+\lambda^{ij}_{t} = \pi^{ij}_{t} S^{i}_{t}/S^{j}_{t} \leq \pi^{ij}_{t}\pi^{ji}_{t}\leq c;
 $$
 that is, the proportional transactions costs are bounded from above. If $X_{t}$ is another element of $K^{*}_{t}$, then
 $$
   X^{j}_{t}\leq X^{i}_{t}\pi^{ij}_{t}=X^{i}_{t}S^{j}_{t}(1+\lambda^{ij}_{t})/S^{i}_{t}\leq c X^{i}_{t}S^{j}_{t}/S^{i}_{t}
 $$
 and thus $X^{j}_{t}/S^{j}_{t}\leq c (X^{i}_{t}/S^{i}_{t})$, which was our assumption in~\eqref{eq:KstarBound}.
\end{example}

\subsection{Fundamental Theorem}

Given $A\subset \R^{d}$, let us write $L^{0}(\cF_{t};A)$ for the set of all $\cF_{t}$-measurable, $A$-valued functions. Similarly, $L^{0}_{P}(\cF_{t};A)$ is the set of all $X\in L^{0}(\cF_{t};\R^{d})$ which are $P$-a.s.\ $A$-valued.

We now introduce a robust version of  the ``no arbitrage of the second kind'' condition; it states that a position which is not solvent today cannot be solvent tomorrow. In the context of transaction costs, this concept was introduced by \cite{Rasonyi.09} under the name ``no sure gain in liquidation value''. More precisely, in the robust version solvency means solvency for all possible models $P\in\cP$,  so let us agree that a property holds \emph{$\cP$-quasi surely} or \emph{$\cP$-q.s.\ }if it holds outside a $\cP$-polar set; that is, a set which is $P$-null for all $P\in\cP$.

\begin{definition}\label{def: NA2}
  We say that $\NA2(\cP)$ holds if for all $t\in\{0,\ldots,T-1\}$ and all $\zeta \in L^{0}(\cF_{t};\R^{d})$,
 $$
 \zeta \in K_{t+1}\;\;\cP\qs \quad\mbox{implies} \quad \zeta \in K_{t}\;\;\cP\qs
 $$
\end{definition}

\begin{remark}\label{rem: local vs global}
 Condition $\NA2(\cP)$ can be formulated equivalently in a global way. To this end, define a \emph{trading strategy} to be an adapted process $\xi$ such that $\xi_{t}\in -K_{t}$ for all $t\in\{0,\ldots,T\}$. This means that the position $\xi_{t}$ can be acquired at time $t$ at no cost, so that $\xi_{t}$ can be seen as the increment of a self-financing portfolio (possibly with consumption). Let $\Xi$ be the set of all trading strategies and consider the following condition: for all $t\in\{0,\ldots,T-1\}$, $\zeta \in L^{0}(\cF_{t};\R^{d})$ and $\xi\in\Xi$,
 \begin{equation}\label{eq:NAprime}
   \zeta+\xi_{t+1}+\cdots + \xi_{T}\in K_{T}\;\;\cP\qs \quad\mbox{implies} \quad \zeta \in K_{t}\;\;\cP\qs \tag{$\NA2(\cP)'$}
 \end{equation}
 To see that $\NA2(\cP)'$ implies $\NA2(\cP)$, it suffices to take $\xi_{t+1}=-\zeta$ $\cP$-q.s.\ and $\xi_{t+2}=\dots=\xi_{T}=0$. Conversely, let $\NA2(\cP)$ hold and $\zeta+\xi_{t+1}+\cdots + \xi_{T}\in K_{T}$ $\cP\qs$ As  $\zeta+\xi_{t+1}+\cdots + \xi_{T-1}\in K_{T}-\xi_{T}\subset K_{T}$ $\cP\qs$, it follows that $\zeta+\xi_{t+1}+\cdots + \xi_{T-1}\in K_{T-1}$ $\cP\qs$ Repeating this argument,  we find that $\zeta \in K_{t}$ $\cP\qs$, so $\NA2(\cP)'$ holds.
\end{remark} 
  
The following condition states that strictly consistent price systems exist for every $P\in\cP$, and more generally that any such system for $[0,t]$ can be extended to $[0,T]$: price systems are consistently extendable, hence the acronym PCE. For $Q\in\fP(\Omega)$, we write $Q\lll \cP$ if there exists $R\in\cP$ such that $Q\ll R$.

\begin{definition}\label{de:PCE} 
	We say that $\PCE(\cP)$ holds if for all $t\in\{0,\ldots,T-1\}$, $P\in \cP$ and $Y\in L^{0}_{P}(\cF_{t},\Int K^{*}_{t})$,
	there exist $Q\in\fP(\Omega)$ and an adapted process $(Z_{s})_{s=t,\dots,T}$ such that
	 \begin{enumerate}
	  \item  $P\ll Q\lll \cP$,
	  \item  $P=Q$ on $\cF_{t}$ and $Y=Z_{t}$ $P$-a.s.
	  \item  $Z_{s}\in \Int K^{*}_{s}$ $Q$-a.s.\ for $s=t,\dots,T$
	  \item  $(Z_{s})_{s=t,\dots,T}$ is a $Q$-martingale.	
  \end{enumerate}
\end{definition}

Conditions (i) and (ii) state that $(Z,Q)$ extends $(Y,P)$, whereas~(iii) and~(iv) state that $(Z,Q)$ is a strictly consistent price system in the sense of \cite{Schachermayer.04}. (The terminology suggests to think of $Y$ as the terminal value of a given consistent price system for $[0,t]$.) 

Let us observe that $Y$ as in Definition~\ref{de:PCE} always exists.

\begin{remark}\label{rk:intSpaceNonempty}
  It follows from~\eqref{eq:assIntKstar} and Lemma~\ref{le:randomSets} that $L^{0}_{P}(\cF_{t};\Int K^{*}_{t})$ is nonempty. More precisely, using Lemma~\ref{le:randomSets} and projecting onto the unit ball, we see that there even exists a bounded $Y\in L^{0}(\cF_{t};\Int K^{*}_{t})$ which is Borel-measurable.
\end{remark}
 
We can now state the main result, a robust fundamental theorem of asset pricing in the spirit of~\cite{Rasonyi.09}. 
 
\begin{theorem}\label{th:FTAP}
  The conditions $\NA2(\cP)$ and $\PCE(\cP)$ are equivalent. 
\end{theorem}

\section{Proof of the Main Result}\label{se:proof}
 
One implication of Theorem~\ref{th:FTAP} is straightforward to prove.

\begin{proof}[Proof that $\PCE(\cP)$ implies $\NA2(\cP)$.]
  Let $t<T$ and let $\zeta \in L^{0}(\cF_{t};\R^{d})$ be such that $\zeta\in K_{t+1}$ $\cP\qs$; we need to show that $\zeta\in K_{t}$ $\cP\qs$ Suppose for contradiction that there is $P\in\cP$ such that $\{\zeta\notin K_{t}\}$ is not $P$-null. Let $\zeta'$ be a Borel-measurable function such that $\zeta'=\zeta$ $P$-a.s.\ (cf.\  \cite[Lemma 7.27]{BertsekasShreve.78}). In view of Lemma~\ref{le:randomSets}, the set
  $$
   \{(\omega,y)\in \Omega_{t}\times \R^{d} : y\in \Int K^{*}_{t}(\omega)\mbox{ and } \br{y,\zeta'(\omega)} <0 \}
  $$
  is Borel. Using measurable selection (Lemma~\ref{le:JvNselectionThm}) and Remark~\ref{rk:intSpaceNonempty}, we may find a universally measurable $Y\in L^{0}(\cF_{t};\Int K^{*}_{t})$ such that $\br{Y,\zeta'}<0$ on $\{\zeta'\notin K_{t}\}$ and in particular 
  \begin{equation}\label{eq:easySideForControd}
   \br{Y,\zeta}<0\;\; P\as \quad\mbox{on}\quad \{\zeta\notin K_{t}\}.
  \end{equation}
  Let $(Z,Q)$ be an extension of $(Y,P)$ as stated in the definition of $\PCE(\cP)$. Using the martingale property of $Z$ and $P\ll Q$ we have
  $$
   0\le E^{Q}[\br{Z_{t+1},\zeta} |\cF_{t}]= \br{Z_{t},\zeta}=\br{Y,\zeta} \quad P\as
  $$
  This contradicts~\eqref{eq:easySideForControd} since $\{\zeta\notin K_{t}\}$ is not $P$-null.
\end{proof}

The reverse implication of Theorem~\ref{th:FTAP} is proved in the remainder of this section.

\subsection{The One-Period Case}\label{se:onePeriod}

In this section, we restrict to the one period model $T=1$. We write $L^{\infty}(\cF_{1};\Int K^{*}_{1})$ for the set of uniformly bounded functions $Y\in L^{0}(\cF_{1};\Int K^{*}_{1})$.

\begin{proposition}\label{pr:oneStep} 
  Let $\NA2(\cP)$ hold. For every $P\in\cP$ and every $y\in \Int K^{*}_{0}$ there exist $P \ll  R\lll \cP$ and $Y\in L^{\infty}(\cF_{1};\Int K^{*}_{1})$ such that $y=E^{R}[Y]$. In particular, $\NA2(\cP)$ implies $\PCE(\cP)$.
\end{proposition}

The following lemma (which is trivial when $\cP$ is a singleton)  will be used in the proof.

\begin{lemma}\label{le:ThetaP}
  For every $P\in\cP$, the set 
  $$
    \Theta_{P}:=   \big\{ E^{R}[Y]:\; P \ll  R\lll \cP,  \; Y\in L^{\infty}(\cF_{1};\Int K^{*}_{1}) \big\} \subset \R^{d}
  $$
  is convex and has nonempty interior.
\end{lemma}

\begin{proof} 
  Let $P\in\cP$. To see the convexity, let  $\alpha\in(0,1)$, $P\ll  R_{i} \lll \cP$ and  $Y_{i}\in L^{\infty}(\cF_{1};\Int K^{*}_{1})$ for $i=1,2$; we show that $\alpha E^{R_{1}}[Y_{1}]+ (1-\alpha)E^{R_{2}}[Y_{2}]$ is in $\Theta_{P}$. Indeed, set $R:=(R_{1}+R_{2})/2$; then $R\lll \cP$ and the densities $H_{i}:= dR_{i}/dR$ satisfy $0\leq H_{i}\leq2$, so that the function 
  $$
    Y:=\alpha  H_{1}Y_{1}+(1-\alpha) H_{2}Y_{2}
  $$
  is bounded. Since $K^{*}_{1}$ is a convex cone, we also have
  $$
    2R\{Y\in \Int K^{*}_{1}\}\ge \sum_{i=1}^{2}R_{i}\{H_{i}Y_{i}\in \Int K^{*}_{1}\} \ge \sum_{i=1}^{2}R_{i}\{H_{i}>0\}=2.
  $$
  After changing $Y$ on an $R$-nullset if necessary, we obtain $Y\in L^{\infty}(\cF_{1};\Int K^{*}_{1})$ and 
  $$
    \alpha E^{R_{1}}[Y_{1}]+ (1-\alpha)E^{R_{2}}[Y_{2}]=E^{R}[Y]\in\Theta_{P}.
  $$
  
  It remains to show that $\Int\Theta_{P}\neq\emptyset$. Let $P \ll  R\lll \cP$ and let $Y\in L^{\infty}(\cF_{1};\Int K^{*}_{1})$ be Borel-measurable; cf.\ Remark~\ref{rk:intSpaceNonempty}. Moreover, let $(e_{i})_{i\le d}$ be a basis of $\R^{d}$. Using Lemma~\ref{le:randomSets}, we see that the set
  $$
  \{(\omega,\eps)\in \Omega_{1}\times (0,1): Y(\omega)\pm\eps e_{i} \in \Int K^{*}_{1}(\omega)\}
  $$
  is Borel for each $i\le d$. Applying Lemma~\ref{le:JvNselectionThm}, we then find universally measurable, $(0,1)$-valued random variables $\epsilon_{i}$ such that $Y\pm\epsilon_{i} e_{i} \in L^{\infty}(\cF_{1};\Int K^{*}_{1})$ for all $i\leq d$.
  We observe that $E^{R}[Y]$ is in the interior of the convex hull of the points $E^{R}[Y\pm\epsilon_{i} e_{i}]$ and thus $E^{R}[Y]\in\Int\Theta_{P}$.
\end{proof}
 
\begin{proof}[Proof of Proposition~\ref{pr:oneStep}] We argue by separation, similarly as~\cite{Rasonyi.09}; in fact, the present argument even seems to be slightly simpler. 
  The claim can be rephrased as $\Int K^{*}_{0}\subset \Theta_{P}$ for all $P\in \cP$. Let $P\in\cP$ and suppose for contradiction that $\Int K^{*}_{0}\not\subset \Theta_{P}$. Let $\bar \Theta_{P}$ denote the closure of $\Theta_{P}$ in $\R^{d}$. By Lemma~\ref{le:ThetaP}, the interior of $\bar \Theta_{P}$ is contained in $\Theta_{P}$ and it follows that $K^{*}_{0}\not\subset \bar\Theta_{P}$. Let $\gamma \in K^{*}_{0}\setminus \bar \Theta_{P}$, then the separating hyperplane theorem and the cone property yield $q\in \R^{d}$ such that 
  \begin{equation}\label{eq:oneStepContra}
   \br{q,\gamma}<0\le \br{q, E^{R}[Y]}=E^{R}[\br{q,Y}]
  \end{equation}
  for all $P\ll R\lll \cP$ and $Y\in L^{\infty}(\cF_{1};\Int K^{*}_{1})$. Let $P'\in\cP$ and $\eps>0$; then this applies in particular to $R:=\eps P +(1-\eps)P'$ and letting $\eps\to0$ yields that
  $$
   0\le E^{P'}[\br{q,Y}],\quad P'\in \cP,\quad Y\in L^{\infty}(\cF_{1};\Int K^{*}_{1}).
  $$ 
  Using dominated convergence, we can strengthen this to
  $$
   0\le E^{P'}[\br{q,Y}],\quad P'\in \cP,\quad Y\in L^{\infty}(\cF_{1};K^{*}_{1}).
  $$
  Applying measurable selection (Lemma~\ref{le:JvNselectionThm}) and Remark~\ref{rk:intSpaceNonempty}, we can find $Y\in L^{\infty}(\cF_{1}; K^{*}_{t})$ such that $\br{q,Y}<0$ on $\{q\notin K_{1}\}$ and $Y=0$ on $\{q\in K_{1}\}$, so it follows that $\{q\notin K_{1}\}$ is $\cP$-polar. By $\NA2(\cP)$, this implies that $q\in K_{0}$, which contradicts the left-hand side of~\eqref{eq:oneStepContra} as $\gamma \in K^{*}_{0}$.
\end{proof}  

\subsection{The Multi-Period Case}

The basic idea in this section is to use the one-period case as a building block for the multi-period case. In the first part of this section, we formulate a local version of  $\NA2(\cP)$ and show that it is implied by $\NA2(\cP)$, up to a polar set. We fix $t\in \{0,\ldots,T-1\}$ until further notice.

\begin{definition}\label{eq: def NA2 local}
  Given $\omega\in\Omega_{t}$, we say that  $\NA2(t,\omega)$ holds if 
  $$
    \zeta\in K_{t+1}(\omega,\cdot) \;\;\cP_{t}(\omega)\qs \quad \mbox{implies} \quad \zeta \in K_{t}(\omega) \quad\mbox{for all $\zeta\in \R^{d}$. }
 $$
\end{definition}

To obtain an analytically tractable description of $\NA2(t,\omega)$, we first introduce a notion of support for the cone $K_{t+1}(\omega,\cdot)$.

\begin{lemma}\label{le:LambdaDef}
	Let $\omega \in \Omega_{t}$. The set 
	$$
	\Lambda_{t}(\omega):=\big\{x\in \R^{d}: x\in K_{t+1}(\omega,\cdot)\;\cP_{t}(\omega)\qs\big\}
	$$
	is a closed convex cone containing the origin. It has the following maximality property: if $A\subset\R^{d}$ is a closed set such that $A\subset K_{t+1}(\omega,\cdot)$ $\cP_{t}(\omega)\qs$, then $A\subset \Lambda_{t}(\omega)$.
\end{lemma}

\begin{proof}
  We fix $\omega$ and omit it in the notation; i.e., we write $\cP_{t}$ for $\cP_{t}(\omega)$, $K_{t+1}$ for $K_{t+1}(\omega,\cdot)$, etc. Moreover, we denote by $B_{\eps}$ the open ball of radius $\eps$ in $\R^{d}$; thus, $C+B_{\eps}$ is the open $\eps$-neighborhood of a set $C\subset \R^{d}$. 
  
  Let $(O_{n})_{n\ge 1}$ be a countable basis of open sets for the topology of $\R^{d}$. We introduce
  $$
	\Lambda_{t}'   :=\bigcap_{k\ge 1}C_{k},\quad \mbox{where} \quad C_{k} :=\bigcup \left\{O_{n}:\, O_{n}\,\subset\, K_{t+1} +B_{1/k}\;\;\cP_{t}\qs\right\}.
	$$
	Since the union is countable, $\Lambda_{t}' \subset K_{t+1} +B_{1/k}$ holds $\cP_{t}\qs$ for all $k\ge1$. As  $K_{t+1}$ is closed, it follows that $\Lambda_{t}'\subset K_{t+1}$ $\cP_{t}\qs$ 
	
	Let $A\subset\R^{d}$ be a closed set such that $A\subset K_{t+1}$ $\cP_{t}\qs$. Then, for all $k\ge 1$, we have $A+B_{1/k}\subset K_{t+1}+B_{1/k}$ $\cP_{t}\qs$, so that the open set $A+B_{1/k}$ is a subset of $C_{k}$. In particular, $A\subset \cap_{k} C_{k}=\Lambda_{t}'$ and we have shown the maximality property for $\Lambda_{t}'$.
	
	Since $\Lambda_{t}'\subset K_{t+1}$ $\cP_{t}\qs$, we have in particular that any $x\in\Lambda_{t}'$ is contained in $\Lambda_{t}$. Conversely, any $x\in\Lambda_{t}$ forms a closed set which is contained in $K_{t+1}$ $\cP_{t}\qs$ and thus $x\in\Lambda_{t}'$ by the maximality. As a result, we have $\Lambda_{t}=\Lambda_{t}'$. That $\Lambda_{t}$ is a closed convex cone containing the origin follows from the fact that $K_{t+1}$ has the same properties.
\end{proof} 

\begin{lemma}\label{le:LambdaMbl}
  The random set $\Lambda_{t}$ is $\cF_{t}$-measurable. 
\end{lemma}
  
\begin{proof}
  Since analytic sets are universally measurable, it suffices to show that the set $\{\Lambda_{t}\cap O\neq\emptyset\}$ is co-analytic whenever $O\subset \R^{d}$ is open. Let $(x_{n})_{n\geq1}$ be dense in $O$, then
  $$
    \{\Lambda_{t}\cap O\neq\emptyset\}=\bigcap_{n\geq1} \{x_{n}\in\Lambda_{t}\}
  $$
  since $\Lambda_{t}$ is the closure of its interior (by Lemma~\ref{le:LambdaDef} and the fact that $\R^{d}_{+}\subset K_{t+1}$). Hence, it suffices to show that $\{x\notin \Lambda_{t}\}$ is analytic for any given point $x\in\R^{d}$. 
  Indeed, by the definition of $\Lambda_{t}$, we have the identity
  \begin{align*}
    \{x\notin \Lambda_{t}\} = \big\{\omega\in\Omega_{t}:\,\exists\, P\in \cP_{t}(\omega)\mbox{ such that } P\{x\notin K_{t+1}(\omega,\cdot)\}>0\big\}.
  \end{align*}
  Lemma~\ref{le:randomSets} and a fact about Borel kernels (cf.\ \cite[Step 1, Proof of Theorem~2.3]{NutzVanHandel.12}) imply that $(P,\omega)\mapsto P\{x\notin K_{t+1}(\omega,\cdot)\}$ is Borel. Using also that the graph of $\cP_{t}$ is analytic, the above identity shows that $\{x\notin \Lambda_{t}\}$ is the projection of an analytic set and thus analytic.
\end{proof}

\begin{lemma}\label{le:NAFailPolar} 
	The set 
	$
	N_{t}:=\{\omega: \NA2(t,\omega)\;\mbox{fails}\}
	$
	is universally measurable. If $\NA2(\cP)$ holds, then $N_{t}$ is $\cP$-polar. 
\end{lemma}
 
\begin{proof} 

  We first show that there exists an $\cF_{t}$-measurable selector $\zeta$ of $\Lambda_{t}$ satisfying 
  $$
    \zeta\in \Lambda_{t}\setminus K_{t}\quad\mbox{on}\quad \{\Lambda_{t}\setminus K_{t}\neq\emptyset\}.
  $$
  Indeed, the random set $\Lambda_{t}$ is closed, nonempty and $\cF_{t}$-measurable; cf. Lemmas~\ref{le:LambdaDef} and~\ref{le:LambdaMbl}. Moreover, the distance function $d(\cdot,K_{t}(\cdot))$ is a Carath\'eodory function with respect to $\cF_{t}$ (even with respect to $\cB(\Omega_{t})$) by Lemma~\ref{le:randomSets}. Thus, the Measurable Maximum Theorem, cf.\ \cite[Theorem~18.19]{AliprantisBorder.06}, yields an $\cF_{t}$-measurable selector $\zeta$ of $\Lambda_{t}\cap \{x\in\R^{d}:\,|x|\leq1\}$ such that
  $$
    d(\zeta,K_{t})  = \max_{x\in \Lambda_{t},\,|x|\leq1} d(x,K_{t}).
  $$
  The right-hand side is strictly positive on $\{\Lambda_{t}\setminus K_{t}\neq\emptyset\}$ and we conclude that $\zeta\in\Lambda_{t}\setminus K_{t}$ on $\{\Lambda_{t}\setminus K_{t}\neq\emptyset\}$ as desired.

  Next, we note that
  $$
   N_{t}=\{\Lambda_{t}\setminus K_{t}\ne \emptyset\}.
  $$
	Indeed, let $\omega\in N_{t}$; that is, $\NA2(t,\omega)$ fails. Then, there exists $x \in\R^{d}\setminus K_{t}(\omega)$ such that $x\in K_{t+1}(\omega,\cdot)$ $\cP_{t}(\omega)\qs$ Thus, $x \in \Lambda_{t}(\omega)$ and $\Lambda_{t}(\omega)\setminus K_{t}(\omega)\ne \emptyset$. Conversely, if there exists some $x\in\Lambda_{t}(\omega)\setminus K_{t}(\omega)$, then $x\in K_{t+1}(\omega,\cdot)$ $\cP_{t}(\omega)\qs$ and $x\notin K_{t}(\omega)$, so $\NA2(t,\omega)$ fails. 
  
  In particular, the above shows that $N_{t}=\{\zeta \in K_{t}\}^{c}$ is universally measurable.
 Moreover, since $\zeta$ is a selector of $\Lambda_{t}$, we have $\zeta(\omega)\in K_{t+1}(\omega,\cdot)$ $\cP_{t}(\omega)\qs$ for all $\omega$ and thus  $\zeta\in K_{t+1}$ $\cP$-q.s.\ by Fubini's theorem. If $\NA2(\cP)$ holds, this implies that $\zeta\in K_{t}$ $\cP\qs$ and thus $N_{t}=\{\zeta \in \Lambda_{t}\setminus K_{t}\}$ is $\cP$-polar.
\end{proof}

Recall that we intend to use the one-period case as a building block for the multi-period case; this will require some measurable selection arguments. However, 
the space $L^{0}(\cF_{1};\R^{d})$ is not directly amenable to measurable selection due to the absence of a reference measure. Hence, we first introduce an auxiliary structure; namely, we shall reformulate the relevant quantities in terms of measures, using that a pair $(Q,g)\in \fP(\Omega_{1})\times L^{1}(Q)$ can be identified with a pair of measures $(Q,Q')$ where $dQ'/dQ=g$. Since we deal with multivariate functions, the following notation will be useful: given a vector $Q=(Q^{i})_{i\le d}\in \fP(\Omega)^{d}$ and a sufficiently integrable random vector $g=(g^{1},\dots,g^{d})\in L^{0}(\cF_{1};\R^{d})$, we set
$$
  E^{Q}[g]:=(E^{Q^{i}}[g^{i}])_{i\le d}.
$$
In a similar spirit, we write $xy$ for the componentwise multiplication (\emph{not} the inner product) of $x,y\in \R^{d}$,
$$
  xy:=(x^{i}y^{i})_{i\le d}, \quad x,y\in \R^{d}. 
$$
Moreover, we set
$$
  D(Q):=(dQ^{i}/dQ^{1})_{i\le d},
$$
where $dQ^{i}/dQ^{1}$ is the Radon--Nikodym derivative of the absolutely continuous part of $Q^{i}$ with respect to $Q^{1}$. We note that the first component of $D(Q)$ equals one. Clearly, the vector $Q$ can be recovered from $Q^{1}$ and $D(Q)$ when $Q$ is an element of
$$
  \cM_{a}:=\{Q\in \fP(\Omega_{1})^{d}: Q^{i}\ll Q^{1} \mbox{ for } i\le d\}.
$$

Using Remark~\ref{rk:intSpaceNonempty}, we choose and fix a Borel-measurable selector $S_{t+1}$ of $\Int K^{*}_{t+1}$. In what follows, it is helpful to think of $S_{t+1}$ as a frictionless price; cf.\ Example~\ref{ex:pi}. Recall that the components of $S_{t+1}$ are strictly positive; dividing by $S^{1}_{t+1}$, we may thus assume that 
\begin{equation}\label{eq:Snormalized}
  S^{1}_{t+1}\equiv1;
\end{equation}
this corresponds to a choice of numeraire.
Given $\omega \in \Omega_{t}$, we define
$$
  \cD_{t}(\omega)= \big\{(Q,\alpha)\in \cM_{a}\times (0,\infty)^{d}: \alpha S_{t+1}(\omega,\cdot) D(Q)\in L^{0}_{Q^{1}}(\cF_{1};\Int K^{*}_{t+1}(\omega,\cdot))\big\}. 
$$
This random set is related to our original goal as follows.

\begin{remark}\label{rk:coding}
	Fix $\omega \in \Omega_{t}$ and $z\in \R^{d}$. Using obvious notation, let $(Y,\tilde Q)\in L^{\infty}_{\tilde Q}(\cF_{1};\Int K^{*}_{t+1}(\omega,\cdot))\times \fP(\Omega_{1})$ be such that $E^{\tilde Q}[Y]=z$, and define the pair $(Q(\omega),\alpha(\omega))\in \fP(\Omega_{1})^{d}\times (0,\infty)^{d}$ by 
	\begin{align*}
	G^{i}(\omega,\cdot)&:=Y^{i}/S^{i}_{t+1}(\omega,\cdot),\\
	\alpha^{i}(\omega)&:=E^{\tilde Q}[G^{i}(\omega,\cdot)], \\
	dQ^{i}(\omega,\cdot)/d\tilde Q&:=\frac{G^{i}(\omega,\cdot)}{\alpha^{i}(\omega)}.
	\end{align*}
	These quantities are well-defined. Indeed, $G^{i}>0$; moreover, recalling~\eqref{eq:Snormalized},
	$$
	  G^{i}(\omega,\cdot)=Y^{i}/S^{i}_{t+1}(\omega,\cdot) \leq cY^{1}/S^{1}_{t+1}(\omega,\cdot) = cY^{1}
	$$
	by our assumption~\eqref{eq:KstarBound}, and since $Y^{1}$ is bounded, $G^{i}$ is bounded and in particular integrable\footnote{
As can be seen here, we can dispense with a condition that is slightly weaker than~\eqref{eq:KstarBound}, but aesthetically less pleasant: the existence of a constant $c$ and selectors $S_{t}$ such that $y^{i}/S^{i}_{t+1}(\omega,\cdot) \leq cy^{1}/S^{1}_{t+1}(\omega,\cdot)$ for all $y\in K^{*}_{t+1}(\omega,\cdot)$. 
}. 
We have $\alpha S_{t+1}(\omega,\cdot) D(Q)=\alpha^{1}Y/Y^{1}\in\Int K^{*}_{t}(\omega,\cdot)$ $\tilde Q$-a.s., so that $(Q(\omega),\alpha(\omega))\in \cD_{t}(\omega)$. Moreover, $\alpha(\omega) E^{Q(\omega)}[S_{t+1}(\omega,\cdot)]=z$.
	
	Conversely, if $(Q,\alpha) \in \cD_{t}(\omega)$ and $\alpha E^{Q}[S_{t+1}(\omega,\cdot)]=z$, then one can set 
	\begin{align*}
	  \tilde Q:=Q^{1} \quad \mbox{and} \quad  Y(\cdot):=\alpha S_{t+1}(\omega,\cdot) D(Q)  
	\end{align*}
	to find $(Y,\tilde Q)\in L^{0}_{\tilde Q}(\cF_{1};\Int K^{*}_{t+1}(\omega,\cdot))\times \fP(\Omega_{1})$ with  $E^{\tilde Q}[Y]=z$. 
\end{remark}  

The next lemma will allow us to select from $\cD_{t}$. We include an additional parameter $R\in\fP(\Omega_{1})$ which will be used to ensure that the consistent price system is dominated by $\cP$.

\begin{lemma}\label{le:Danalytic} 
	Let  $P(\cdot): \Omega_{t}\to \fP(\Omega_{1})$ be a Borel kernel. The set of all 
	$(\omega,z, \alpha,Q,R)\in \Omega_{t}\times \R^{d}\times \R^{d}\times \fP(\Omega_{1})^{d}\times \fP(\Omega_{1})$ such that 
	$$
	\alpha E^{Q}[S_{t+1}(\omega,\cdot)]=z,\quad P(\omega)\ll  Q^{1}\ll R,\quad R\in \cP_{t}(\omega),\quad (Q,\alpha)\in \cD_{t}(\omega)
	$$  
	is analytic. 
\end{lemma}

\begin{proof} 
  Let $\Gamma$ be the set of $(\omega, z,\alpha,Q,R)\in  \Xi:= \Omega_{t}\times \R^{d}\times \R^{d} \times \fP(\Omega_{1})^{d}\times \fP(\Omega_{1})$ such that  
  $$
   \alpha E^{Q}[S_{t+1}(\omega,\cdot)]=z,\quad P(\omega)\ll  Q^{1}\ll R,\quad (Q,\alpha)\in \cD_{t}(\omega).
  $$
  Then the set in question is the intersection 
  $$
    \Gamma \cap \{(\omega,z, \alpha,Q,R)\in \Xi : R\in \cP_{t}(\omega)\}.
  $$  
  The second set is analytic because $\graph(\cP_{t})$ is analytic; thus, the lemma will follow if $\Gamma$ is analytic as well.  
  We show that $\Gamma$ is even Borel.  Observe first that the map  
 $$
  (\omega,z,\alpha,Q,R)\in \Xi \mapsto  \alpha E^{Q}[S_{t+1}(\omega,\cdot)]-z
 $$
 is Borel (again, using \cite[Step 1, Proof of Theorem 2.3]{NutzVanHandel.12}), so that
 $$ 
   \Gamma_{1}:=\big\{(\omega, z,\alpha,Q,R)\in \Xi: \alpha E^{Q}[S_{t+1}(\omega,\cdot)]=z \big\}
 $$
is Borel. A classical fact about Radon--Nikodym derivatives on separable spaces (e.g., \cite[Lemma 4.7]{BouchardNutz.13}) shows that $(Q,\tilde\omega)\mapsto D(Q)(\tilde\omega)$ can be chosen to be jointly Borel. Using this result, one can show that the set 
 $$
   \Gamma_{2}:=\big\{(\omega, z,\alpha,Q,R)\in \Xi:  \; Q\in \cM_{a},\; P(\omega)\ll  Q^{1}\ll R\big\}
 $$
 is also Borel; the proof follows similar arguments as \cite[Lemma 4.8]{BouchardNutz.13}. Moreover, 
 $$
  (\omega, \tilde \omega,\alpha,Q)\in \Xi':=\Omega_{t}\times \Omega_{1}\times\R^{d}\times \fP(\Omega_{1})^{d}\mapsto \alpha S_{t+1}(\omega,\tilde \omega)  D(Q)(\tilde \omega)
 $$
 is Borel. Now Lemma~\ref{le:randomSets} implies that  
 $$
A:=\big\{(\omega,\tilde \omega, \alpha,Q)\in \Xi' :  \alpha S_{t+1}(\omega,\tilde \omega)D(Q)(\tilde \omega)\in  \Int K^{*}_{t+1}(\omega,\tilde \omega) \big\}
 $$
 is Borel and then so is
 $$
 \phi:(\omega,z, \alpha,Q,R)\in \Xi \mapsto E^{R}\bigg[\frac{dQ^{1}}{dR} \1_{A}(\omega,\cdot, \alpha,Q)\bigg].
 $$ 
 This shows that 
\begin{align*}
 \Gamma_{3}:=\big\{(\omega,z, \alpha,Q,R)\in \Xi: \phi(\omega,z, \alpha,Q,R)\ge 1 \big\}
\end{align*}
is Borel. As a result, $\Gamma=\Gamma_{1}\cap\Gamma_{2}\cap\Gamma_{3}$ is Borel as claimed.
\end{proof} 

We can now complete the proof of Theorem~\ref{th:FTAP}.

\begin{lemma}
  $\NA2(\cP)$ implies $\PCE(\cP)$.
\end{lemma}

\begin{proof}
  Let $t\in\{0,\ldots,T-1\}$, $P\in \cP$ and $Y\in L^{0}_{P}(\cF_{t},\Int K^{*}_{t})$; by choosing a suitable version we may assume that $Y$ is Borel. We first construct the extension to time $t+1$.
  
  The measure $P$ is of the form $P=P_{t-1}\otimes P_{t}(\cdot) \otimes \cdots \otimes P_{T-1}(\cdot)$, where $P_{t-1}=P|_{\cF_{t}}$ and   $P_{t}(\cdot): \Omega_{t}\to \fP(\Omega_{1})$ is a Borel kernel such that $P_{t}(\omega)\in\cP_{t}(\omega)$ for $P_{t-1}$-a.e.\ $\omega$, and similarly for $P_{t+1}(\cdot), \dots, P_{T-1}(\cdot)$.
  
  Let $\omega\in\Omega_{t}$ be such that $Y(\omega)\in\Int K^{*}_{t}(\omega)$, $P_{t}(\omega)\in\cP_{t}(\omega)$ and $\NA2(t,\omega)$ holds; these $\omega$ form a set of full $P_{t-1}$-measure by Lemma~\ref{le:NAFailPolar}. Then, the one-step result of Proposition~\ref{pr:oneStep} and Remark~\ref{rk:coding} show that there exist $(Q_{t}(\omega),\alpha(\omega),R_{t}(\omega))\in\cD_{t}(\omega)\times\cP_{t}(\omega)$ such that 	
  $$
	  \alpha(\omega) E^{Q_{t}(\omega)}[S_{t+1}(\omega,\cdot)]=Y(\omega) \quad\mbox{and}\quad P_{t}(\omega)\ll  Q^{1}_{t}(\omega)\ll R_{t}(\omega).
	$$
	Moreover, it follows from Lemma~\ref{le:Danalytic} that the map $\omega\mapsto (Q_{t}(\omega),\alpha(\omega), R_{t}(\omega))$ can be chosen in a universally measurable way; cf.\ Lemma~\ref{le:JvNselectionThm}. Using the opposite direction of Remark~\ref{rk:coding}, setting $Z_{t+1}(\omega,\cdot):=\alpha(\omega) S_{t+1}(\omega,\cdot) D(Q_{t}(\omega))$ yields 
	$$
	  Z_{t+1}(\omega,\cdot)\in L^{0}_{Q_{t}^{1}(\omega)}(\cF_{1};\Int K^{*}_{t+1}(\omega,\cdot)) \quad \mbox{with}\quad  E^{Q_{t}^{1}(\omega)}[Z_{t+1}(\omega,\cdot)]=Y(\omega). 
 $$
 We see from its definition that $Z_{t+1}$ is $\cF_{t+1}$-measurable. Moreover, Fubini's theorem shows that 
 \begin{align*}
   Q&:=P_{t-1}\otimes Q_{t}^{1}(\cdot) \otimes P_{t+1}(\cdot) \otimes \cdots \otimes P_{T-1}(\cdot),\\
   R&:=P_{t-1}\otimes R_{t}(\cdot) \otimes P_{t+1}(\cdot) \otimes \cdots \otimes P_{T-1}(\cdot)
 \end{align*}
 are measures satisfying $P\ll Q\ll R \in \cP$ and $P=Q$ on $\cF_{t}$ as well as $E^{Q}[Z_{t+1}|\cF_{t}]=Y=:Z_{t}$ $Q$-a.s.; here the last conclusion uses that the components of $Z$ are nonnegative. Thus, we have constructed the desired extension to time $t+1$.
 
 We may iterate this argument, using $Q$ instead of $P$ and $Z_{t+1}$ instead of~$Y$, to find the required extension of $(P,Y)$ up to time $T$.
\end{proof}

\appendix
\section{Appendix}

\subsection{Measure Theory}

Given a measurable space $(\Omega,\cA)$, let $\fP(\Omega)$ the set of all probability measures on $\cA$. The \emph{universal completion} of $\cA$ is the $\sigma$-field $\cap_{P\in\fP(\Omega)} \cA^P$, where $\cA^P$ is the $P$-completion of $\cA$. When $\Omega$ is a topological space with Borel $\sigma$-field $\cB(\Omega)$, we always endow $\fP(\Omega)$ with the topology of weak convergence. Suppose that $\Omega$ is Polish, then $\fP(\Omega)$ is also Polish.  A subset $A\subset\Omega$ is \emph{analytic} if it is the image of a Borel subset of another Polish space under a Borel-measurable mapping. Analytic sets are stable under countable union and intersection, under forward and inverse images of Borel functions, but not under complementation: the complement of an analytic set is called \emph{co-analytic} and it is not analytic unless it is Borel. Any Borel set is analytic, and any analytic set is universally measurable; i.e., measurable for the universal completion of~$\cB(\Omega)$. We refer to \cite[Chapter~7]{BertsekasShreve.78} for these results and further background.

\subsection{Random Sets}

Let $(\Omega,\cA)$ be a measurable space. A mapping $\Psi$ from $\Omega$ into the power set~$2^{\R^{d}}$ will be called a \emph{random set} in $\R^{d}$ and its \emph{graph} is defined as
\[
  \graph(\Psi)=\{(\omega,x):\, \omega\in\Omega,\, x\in \Psi(\omega)\}\subset \Omega\times \R^{d}.
\]
We say that $\Psi$ is \emph{$\cA$-measurable} (\emph{weakly $\cA$-measurable}) if
$$
  \{\omega\in\Omega:\, \Psi(\omega)\cap O\neq \emptyset\}\in \cA\quad \mbox{for all closed (open) $O\subset \R^{d}$.}
$$
Moreover, $\Psi$ is called closed (convex, etc.) if $\Psi(\omega)$ is closed (convex, etc.) for all $\omega\in\Omega$. We emphasize that measurability is \emph{not} defined via the measurability of the graph, as it is sometimes done in the literature.

\pagebreak[2]

\begin{lemma}\label{le:randomSets}
  Let $(\Omega,\cA)$ be a measurable space and let $\Psi$ be a closed, nonempty random set in $\R^{d}$. The following are equivalent:
  \begin{enumerate}
    \item  $\Psi$ is $\cA$-measurable.
    \item  $\Psi$ is weakly $\cA$-measurable.
    \item The distance function $d(\Psi,y)=\inf\{x\in \Psi:\, |x-y|\}$ is $\cA$-measurable for all $y\in\R^{d}$.     
   \item There exist $\cA$-measurable functions $(\psi_{n})_{n\geq1}$ such that $\Psi=\overline{\{\psi_{n}, n \geq 1\}}$ (``Castaing representation'') .
  \end{enumerate}
  Moreover, (i)--(iv) imply that
  \begin{enumerate}
    \item[(v)]  $\graph(\Psi)$ is $\cA\times\cB(\R^{d})$-measurable.
    \item[(vi)]  The dual cone $\Psi^{*}$ is $\cA$-measurable.
    \item[(vii)] $\graph(\Int \Psi^{*})$ is $\cA\times\cB(\R^{d})$-measurable.
    \item[(viii)] There exists an $\cA$-measurable selector $\psi$ of $\Psi^{*}$ satisfying $\psi\in \Int \Psi^{*}$ on $\{\Int \Psi^{*}\neq\emptyset\}$.  
\end{enumerate}
  If $\cA$ is universally complete, then (v) is equivalent to (i)--(iv).
\end{lemma}

\begin{proof}
  We refer to \cite{Rockafellar.76} for the results concerning (i)--(vi). Let (iv) hold, then we have the representation 
  \begin{align*}
    \graph(\Int \Psi^{*})
    &=\{(\omega,y)\in\Omega\times\R^{d}:\, \br{x,y}>0\mbox{ for all } x\in \Psi(\omega)\setminus\{0\}\}\\
    &=\bigcap_{n\geq1}\{(\omega,y)\in\Omega\times\R^{d}:\, \br{\psi_{n}(\omega),y}>0 \mbox{ or }\psi_{n}(\omega)=0\}
  \end{align*}
  which readily implies (vii). 
  
  
  Finally, let (vi) hold, then $\Psi^{*}$ has a Castaing representation $(\phi_{n})$. Let $\phi:=\sum_{n} 2^{-n} \phi_{n}$, then $\phi$ is $\cA$-measurable and $\Psi^*$-valued since $\Psi^*$ is closed and convex. Let $\omega\in\Omega$ be such that $\Int \Psi^*(\omega)\neq\emptyset$. By the density, at least one of the points $\phi_{n}(\omega)\in \Psi^*(\omega)$ must lie in the interior of $\Psi^*(\omega)$. Moreover, since $\Psi^*(\omega)$ is convex, we observe that a nondegenerate convex combination of a point in $\Psi^*(\omega)$ with an interior point of $\Psi^*(\omega)$ is again an interior point. These two facts yield that $\phi(\omega)\in \Int \Psi^*(\omega)$ as desired. (This applies to any closed and nonempty convex random set, not necessarily of the form~$\Psi^{*}$.)
\end{proof}

In some cases we need to select from random sets in infinite-dimensional spaces, or random sets that are not closed. The following is sufficient for our purposes.

\begin{lemma}[Jankov--von Neumann]\label{le:JvNselectionThm}
  Let $\Omega, \Omega'$ be Polish spaces and let $\Gamma\subset \Omega\times\Omega'$ be an analytic set. Then the projection $\pi_{\Omega}(\Gamma)\subset \Omega$ is universally measurable and there exists a universally measurable function $\psi:\pi_{\Omega}(\Gamma)\to \Omega'$ whose graph is contained in $\Omega'$.
\end{lemma}

We refer to \cite[Proposition 7.49]{BertsekasShreve.78} for a proof. In many applications we start with a random set $\Psi:\Omega\to 2^{\Omega'}$ such that $\Gamma:=\graph(\Psi)$ is analytic. Noting that $\pi_{\Omega}(\Gamma)=\{\Psi\neq\emptyset\}$, Lemma~\ref{le:JvNselectionThm} then yields a universally measurable selector for $\Psi$ on the set $\{\Psi\neq\emptyset\}$.


\newcommand{\dummy}[1]{}


\end{document}